\documentclass[11pt]{article}
\usepackage{amsmath,amsthm,amssymb,amsfonts}
\usepackage[colorlinks]{hyperref}
\hypersetup{linkcolor=blue, citecolor=blue}
\usepackage{cleveref}
\usepackage{graphicx}
\usepackage{fullpage}
\usepackage{float}
\usepackage{tikz}
\usepackage{forest}
\usetikzlibrary{quantikz2}
\usepackage{appendix}
\usepackage[normalem]{ulem}
\usepackage{thm-restate}
\usepackage{tcolorbox}
\usepackage{enumitem}

\allowdisplaybreaks

\Crefformat{equation}{(#2#1#3)}

\usepackage{complexity}
\newcommand{\QTC}{\ComplexityFont{QTC}}
\newcommand{\EQTC}{\ComplexityFont{EQTC}}
\newcommand{\BQTC}{\ComplexityFont{BQTC}}
\newcommand{\BQAC}{\ComplexityFont{BQAC}}
\newcommand{\EQAC}{\ComplexityFont{EQAC}}
\newcommand{\EQNC}{\ComplexityFont{EQNC}}

\newcommand{\Copy}{\mathsf{Copy}}

\newtheorem{theorem}{Theorem}

\newtheorem{claim}[theorem]{Claim}

\newtheorem{corollary}[theorem]{Corollary}

\newtheorem{lemma}[theorem]{Lemma}

\newtheorem{remark}[theorem]{Remark}

\newcommand{\AND}{\mathsf{AND}}
\newcommand{\OR}{\mathsf{OR}}

\renewcommand{\MOD}{\mathsf{MOD}}

\newcommand{\EX}{\mathsf{EX}}

\newcommand{\Maj}{\operatorname{Maj}}
\newcommand{\Ind}{\mathsf{Ind}}
\newcommand{\Sym}{\mathsf{Sym}}

\newcommand{\Fanout}{\mathsf F} 

\newcommand{\Threshold}{\mathsf{Th}} 
\newcommand{\Parity}{\mathsf{Parity}} 

\newcommand{\sket}[1]{\,|#1\rangle}

\renewcommand{\epsilon}{\varepsilon}

\begin{document}

\title{$\mathsf{QAC}^0$ Contains $\mathsf{TC}^0$ (with Many Copies of the Input)}

\author{
Daniel Grier\thanks{UCSD. Email: \texttt{dgrier@ucsd.edu}}
\and 
Jackson Morris\thanks{UCSD. Email: \texttt{jrm035@ucsd.edu}}
\and
Kewen Wu\thanks{Institute for Advanced Study. Email: \texttt{shlw\_kevin@hotmail.com}.}
}
\date{}
\maketitle

\begin{abstract}
$\mathsf{QAC}^0$ is the class of constant-depth polynomial-size quantum circuits constructed from arbitrary single-qubit gates and generalized Toffoli gates. 
It is arguably the smallest natural class of constant-depth quantum computation which has not been shown useful for computing \emph{any} non-trivial Boolean function. Despite this, many attempts to port classical $\mathsf{AC}^0$ lower bounds to $\mathsf{QAC}^0$ have failed. 

We give one possible explanation of this: $\mathsf{QAC}^0$ circuits are significantly more powerful than their classical counterparts.
We show the unconditional separation $\mathsf{QAC}^0\not\subset\mathsf{AC}^0[p]$ for \emph{decision} problems, which also resolves for the first time whether $\mathsf{AC}^0$ could be more powerful than $\mathsf{QAC}^0$.
Moreover, we prove that $\mathsf{QAC}^0$ circuits can compute a wide range of Boolean functions if given multiple copies of the input: $\mathsf{TC}^0 \subseteq \mathsf{QAC}^0 \circ \mathsf{NC}^0$. 
Along the way, we introduce an amplitude amplification technique that makes several approximate constant-depth constructions exact.
\end{abstract}

\tableofcontents
\newpage

\section{Introduction}\label{sec:intro}

Constant-depth quantum circuits have long played a central role in our understanding of how quantum computers can gain an advantage over their classical counterparts. Indeed, even constant-depth circuits consisting entirely of single- and two-qubit gates can provably outperform classical circuits at a variety of sampling and searching tasks \cite{bravyi2018quantum, watts2019exponential, grier2020interactive, watts2023unconditional,kane2024locality, sampqnc0}. However, as soon as we consider \emph{decision} problems (where to goal is to compute a single output bit of a Boolean function), the story changes dramatically. The small light cones of such constant-depth quantum circuits significantly constrain their behavior, making them no more powerful than constant-depth classical circuits. This phenomenon manifests as the complexity class equality $\QNC^0 = \NC^0$.

In other words, to witness the power of constant-depth quantum circuits for computing Boolean functions, the quantum circuit must have access to large entangling gates that can act on many qubits at once. One particularly important class of constant-depth polynomial-size quantum circuits is $\QAC^0$, where the circuit can apply arbitrary single-qubit gates as well as generalized Toffoli gates (i.e., the reversible $n$-bit $\AND$ function). $\QAC^0$ has garnered recent attention as a possible viable model from some near-term quantum hardware \cite{wang2001multibit}, and moreover, it has a long history of study as the natural analog to the famous classical circuit class $\AC^0$. While it was known that $\AC^0$ has quite limited computational power, it has been much harder to show limitations on the power of $\QAC^0$. 

Perhaps the most famous and illustrative example of this discrepancy is witnessed by the parity function. While there are several techniques for showing that parity is not computable in $\AC^0$, $\QAC^0$ has survived a long line of research seeking to prove the same result in the quantum world \cite{lb_fanout,depth2,rosenthal, pauli_spec,anshu_dong_ou_yao:2024_QAC0,bera_lb,depth3}. Critically, the same techniques that allow for lower bounds in the classical world, like random restrictions \cite{ajtai_parity, fss_ac0, hastad_thesis} or Fourier concentration \cite{lmn}, have failed to port over to the quantum world (at least if you do not restrict the use of ancillas).
Of course, one possible explanation of this phenomenon is that $\QAC^0$ are just significantly more powerful than previously assumed.

With this in mind, it is natural to search for $\QAC^0$ circuits which might exemplify this power. Until recently, this approach has received relatively little attention. The first nontrivial constant-depth $\QAC$ circuit\footnote{We will use $\QAC$ to refer to quantum circuits consisting of single-qubit and generalized Toffoli gates. That is, $\QAC^0$ is the class of $\QAC$ circuits of constant depth and polynomial size.} construction was given by Rosenthal \cite{rosenthal} where generalized Toffoli gates of exponential size are leveraged in constant depth to approximately compute parity and the quantum fanout\footnote{Quantum fanout is the following operation on classical basis states: $\ket{b, x_1, \ldots, x_n} \mapsto \ket{b, x_1 \oplus b, \ldots, x_n \oplus b}$.} gate.
More recently, it was shown that weak pseudorandom unitaries can be implemented with $\QAC^0$ circuits \cite{foxman2025random}. One ingredient in this construction involves shrinking the exponential-size circuit of \cite{rosenthal} to compute fanout on logarithmically many qubits using a $\QAC^0$ circuit of polynomial size. We will also make use of this technique in this paper.
Nevertheless, it is unclear how one might extend this construction to implement fanout on a larger number of qubits. This lack of large fanout serves as a potential barrier for certain circuit construction techniques.
For example, one can show that large fanout is necessary for classical $\AC^0$ circuits to compute even relatively simple Boolean functions such as the indexing function (see \Cref{sec:index_fanout} for a proof).

In light of this, it is natural to search for inherently quantum primitives that can be constructed with generalized Toffoli gates to compute non-trivial Boolean functions.
This search is exactly the focus of our work. 
In particular, we give a decision problem which can be solved by a $\QAC^0$ circuit, but requires exponential-size $\AC^0[p]$ circuits.

\begin{theorem}[See also \Cref{thm:qac0_sep}]\label{thm:intro_sep_thm}
    There exists a language $L$ which can be decided by a $\QAC^0$ circuit with perfect completeness and soundness $2^{-\poly(n)}$ on inputs of size $n$. However, $L$ requires $\AC^0[p]$ circuits of size $2^{\poly(n)}$ for all primes $p > 1$. Thus, $\BQAC^0 \not \subset \AC^0[p]$.
\end{theorem}

We note that $\AC^0[p]$ strictly contains $\AC^0$ \cite{ajtai_parity, fss_ac0, hastad_thesis}, so \Cref{thm:intro_sep_thm} immediately implies the novel separation $\QAC^0\not\subset\AC^0$. In other words, prior to our result, it was conceivable that $\AC^0$ was \emph{strictly more powerful} than $\QAC^0$ for solving decision problems. We refute this possibility, showing that constant-depth quantum circuits can implement hard-to-compute Boolean functions, even without large fanout.

It turns out that \Cref{thm:intro_sep_thm} follows from what is perhaps an even more surprising aspect of $\QAC^0$ circuits. Namely, we show that providing $\QAC^0$ circuits with polynomially many copies of the input string allows them to simulate arbitrary $\TC^0$ computations.\footnote{Recall that $\TC^0$ is the set of languages that can be computed with constant-depth polynomial-size threshold circuits. $\TC^0$ strictly contains $\AC^0[p]$ \cite{razborov,smolensky}.}

\begin{theorem}[See also \Cref{thm:tc0_in_qac0_copy}]\label{thm:poly_copies}
     Any $\TC^0$-computable function can be decided in $\QAC^0$ with bounded error and polynomially many copies of the input. Equivalently, $\TC^0 \subseteq \BQAC^0 \circ \NC^0$.
\end{theorem}

While \Cref{thm:intro_sep_thm} gives the existence of a single language separating $\QAC^0$ circuits from $\AC^0[p]$ circuits, \Cref{thm:poly_copies} shows that this separation comes from the fact that $\QAC^0$ are generically powerful. As one example, since $\TC^0$ circuits can multiply $n$-bit integers \cite{hesse2002uniform}, then so can $\QAC^0$ circuits (at least when provided multiple copies of the input). This resolves the question of whether multiple classical copies enable $\QAC^0$ to compute non-trivial Boolean functions, a question posed by Rosenthal in his thesis \cite{rosenthal_thesis}.

The circuit constructions leading to the two theorems above follow from two key primitives: the  ``W test'' for Hamming weight detection (see \Cref{sec:overview}) and ``exact amplitude amplification''. The latter method immediately allows us to remove the approximation error present in several prior constant-depth constructions:

\begin{corollary}[See also \Cref{cor:exact_qac0}]\label{cor:qac0_exact_intro}
Parity can be exactly computed by constant-depth exponential-size $\QAC$ circuits.
\end{corollary}
\begin{corollary}[See also \Cref{cor:exact_qtc0}]\label{cor:qtc0_exact_intro}
 Parity can be exactly computed by constant-depth polynomial-size $\QTC$ circuits. Consequently, $\QTC^0 = \QNC^0_\mathrm{wf} = \QAC^0_\mathrm{wf}$.
\end{corollary}
 
\Cref{cor:qac0_exact_intro} and \Cref{cor:qtc0_exact_intro} resolve open questions of \cite{rosenthal_thesis} and \cite{grier_morris} respectively.
In fact, exact amplitude amplification also allows us to refine a construction of \cite{anshu_dong_ou_yao:2024_QAC0} to produce nice states with long-range entanglement in $\QAC^0$. In particular, let
\[
    \ket W = \frac{1}{\sqrt n} \sum_{|x| = 1} \ket{x}
\]
be the uniform superposition of $n$-bit strings of Hamming weight exactly one.
 
\begin{theorem}[See also \Cref{thm:exact_w_1}]\label{thm:exact_w}
There exists a $\QAC^0$ circuit $U$ such that $U\ket{0^n}\ket{0^a}=\ket W\ket{0^a}$ where $a=\poly(n)$.
\end{theorem}

\subsection{Technical Overview}\label{sec:overview}

We give an overview of our techniques.

\paragraph{The $W$ Test.} 
The primary technical tool underlying \Cref{thm:intro_sep_thm} and \Cref{thm:poly_copies} is a quantum primitive we call the \emph{$W$ test}.

\begin{center}
\begin{tcolorbox}[standard jigsaw, coltitle=black, colbacktitle=black!05, arc=5pt, boxrule=.7pt, opacityback=0, title={The $W$ test for detecting Hamming weight $n/2$.}]
\begin{tabular}{rl}
\textbf{Required:} & Unitary $U$ for preparing the $\ket W$ state: $U\ket{0^n} = \ket{W}$ (ancillas ommited)\\ [4pt]
\textbf{Goal:} & For input $x \in \{0,1\}^n$, compute $\EX_{n/2}(x)$, i.e., check if $|x| = n/2$. \\ [4pt]
\textbf{Error:} & If $|x| = n/2$, accept with certainty. \\
                & If $|x| \neq n/2$, reject with probability $1/n^2$.
\end{tabular}
\end{tcolorbox}
\end{center}

Roughly, the $W$ test makes uses of the $W$ state preparation unitary to ``weakly'' compute the exact function $\EX_{n/2}$, which is $1$ iff the Hamming weight of the input is exactly $n/2$.
Specifically, if $U$ is some $n$-qubit unitary which prepares the $\ket W$ state (we will come back to this state preparation task soon), then the $W$ test uses $U$ and its inverse in constant depth as depicted in \Cref{fig:w-test} below.

\begin{figure}[H]
\centering
\scalebox{0.8}{\input{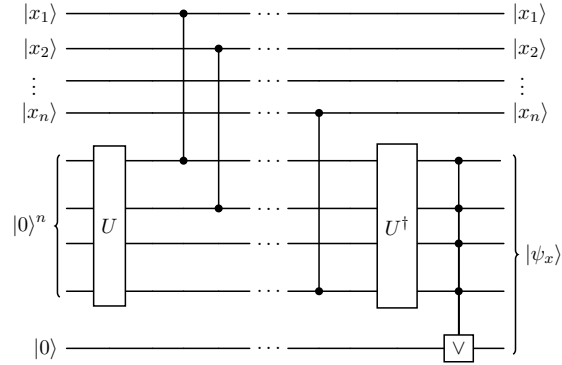}}
\caption{The $W$ test circuit which weakly computes $\EX_{n/2}$, where the two-qubit gates in the second layer above are controlled $Z$-gates.}\label{fig:w-test}
\end{figure}

The circuit in \Cref{fig:w-test} is in fact identical to Moore's construction \cite{moore:1999} except in our case $U$ prepares the $\ket W$ state rather than the cat state. Moore's circuit shows that the cat state and the parity function are dual: the cat state may be used to compute parity with low overhead and vis-a-versa. While approximate, the $W$ test suggests that a similar duality exists between $\ket{W}$ and the $\EX_{n/2}$ function. 

The analysis of the $W$ test is a straightforward calculation (see \Cref{lem:w_test} for a formal proof): 
$$
\ket{\psi_x}=\frac{n-2|x|}n\ket{0^n}\ket0+\sqrt{1-\left(\frac{n-2|x|}n\right)^2}\ket{*_x}\ket1,
$$
where $\ket{*_x}$ is a normalized state depending on $x$.
This immediately implies that for $|x|=n/2$, the last qubit of $\ket{\psi_x}$ is always $\ket1$; whereas if $|x|\ne n/2$, the last qubit of $\ket{\psi_x}$ measures to $\ket0$ with probability at least $1/n^2$.

\paragraph{Quantum-Classical Separations.}
To obtain \Cref{thm:intro_sep_thm} and \Cref{thm:poly_copies}, we note an important feature of the $W$ test: one-sided error. 
Indeed, if $\EX_{n/2}(x)=1$, then the $W$ test is always correct; otherwise $\EX_{n/2}(x)=0$, and it is correct with probability at least $1/n^2$.
This means we can perform error reduction using the $\AND$ function (which notably is in $\QAC^0$), instead of using the majority function (which is not known to be in $\QAC^0$).

More formally, let us take the $\AND$ of $c\approx n^2$ parallel runs of the $W$ tests.
If $\EX_{n/2}(x)=1$, then the $\AND$ outcome is always $1$; otherwise the outcome is $0$ with probability $0.99$.
This is \emph{almost} a proof of \Cref{thm:intro_sep_thm} as $\EX_{n/2}$ is not in $\AC^0[p]$ \cite{razborov,smolensky}.

The caveat here is, $c$ parallel runs of the $W$ tests require $c$ disjoint copies of the input (see \Cref{sec:better_copy_complexity} for partial progress to bypass this) and it is unclear how to make this many copies of the input in $\QAC^0$.
Fortunately, to prove separation results, we do not have to stick to the vanilla $\EX_{n/2}$ function. We can use a variant.
In particular, we define the $c$-copy version of $\EX_{n/2}$ as $\Copy\EX_{n/2}\colon \{0,1\}^{n \times c} \to\{0,1\}$:
$$
\Copy\EX_{n/2}(X_1,\ldots,X_c)=\begin{cases}
\EX_{n/2}(X_1) & X_1=\cdots=X_c,\\
0 & \text{otherwise,}
\end{cases}
\quad\text{for every $X_1,\ldots,X_c\in\{0,1\}^n$.}
$$
The benefit of using $\Copy\EX_{n/2}$ is two-fold.
\begin{itemize}
\item Regarding classical lower bounds, $\AC^0[p]$ circuits can freely make copies of the input as unbounded fanout is permitted. Hence $\Copy\EX_{n/2}\notin\AC^0[p]$ follows immediately from the known lower bound $\EX_{n/2}\notin\AC^0[p]$. 
\item Regarding quantum upper bounds, $\QAC^0$ circuits can now use the $c$ provided input copies to execute parallel runs of the $W$ test. In addition, the consistency check among input copies can be done using generalized Toffoli gates in parallel for all $n$ coordinates. This implies that $\Copy\EX_{n/2}\in\QAC^0$.
\end{itemize}
This completes the proof of \Cref{thm:intro_sep_thm}.
For \Cref{thm:poly_copies}, it suffices to show how to compute all threshold functions with $\QAC^0$ circuits given many copies of the input string. This is a standard padding argument that reduces threshold functions to $\EX_{n/2}$ and then use the circuit for $\Copy\EX_{n/2}$.
See \Cref{sec:sep_qac0_lightweight} for details.

\paragraph{Exact Preparation of $\ket W$ and Nekomata.}
Now we come back to the construction of the $\ket W$ state (\Cref{thm:exact_w}), as required in \Cref{fig:w-test} and the separations above.
To exactly prepare $\ket W$ with a $\QAC^0$ circuit, we rely on \emph{exact amplitude amplification} \cite{grover1998quantum,brassard2002quantum}. 
In general, amplitude amplification cannot be performed in constant depth \cite{zalka1999grover}. The restricted setting considered here consists of (1) a ``warm-start'' state which non-trivially approximates the target state and (2) a method to ``flag'' the target state. 
Standard amplitude amplification can then be performed to obtain the \emph{exact} target state (see \Cref{thm:flag_grover} for more details).
While this technique has been previously considered for state preparation tasks \cite{t_count, rosenthal2024efficient}, they have not yet been applied to the study of $\QAC^0$ circuits.

For the $\ket W$ state, the ``warm-start'' state turns out to be the product state $(\delta\ket1+\sqrt{1-\delta^2}\ket0)^{\otimes n}$ where $\delta\approx1/\sqrt n$.
In addition, the ``flag'' procedure corresponds to computing the exact threshold function $\EX_1$ to filter out strings of Hamming weight exactly one.
We note that, at this point, it should \emph{not} be obvious how to compute $\EX_1$, and we will address this shortly.

Aside from the $\ket W$ state, the above protocol also makes exact several approximate constructions of \emph{the nekomata state}, which is a family of states of significant importance to constant-depth circuits \cite{qacc,rosenthal}.
A quantum state $\ket{\psi}$ is a nekomata if $\ket{\psi}=\frac1{\sqrt2}\ket{0^n}\ket{\phi_0} + \frac1{\sqrt2}\ket{1^n}\ket{\phi_1}$ for some normalized states $\ket{\phi_0}$ and $\ket{\phi_1}$. 
To obtain exact nekomata, the ``warm-start'' states are highly nontrivial but fortunately provided by \cite{rosenthal_thesis,grier_morris}; and the ``flag'' procedure turns out to be a simple $\QAC^0$ filtering of all-zero $\ket{0^n}$ and all-one $\ket{1^n}$ strings. See \Cref{sec:exact_aa} for details.

\paragraph{Exact Computation in $\QAC^0$.}
It is known that computing parity and implementing quantum fanout are equivalent to constructing nekomata states \cite{rosenthal_thesis,grier_morris}.
Hence the above \emph{exact} preparation of nekomata leads to an \emph{exact} computation of the parity function.
In the context of $\QAC^0$, this proves \Cref{cor:qac0_exact_intro}; and in the context of $\QTC^0$, this proves \Cref{cor:qtc0_exact_intro}.

By shrinking the construction of \Cref{cor:qac0_exact_intro} and the equivalence between parity and fanout, the exact fanout gate of $\polylog(n)$ size can be synthesized by constant-depth $\poly(n)$-size $\QAC$ circuits.
This immediately allows $\QAC^0$ to simulate $\AC^0$ circuits of polylogarithmic fanout in an exact way.
In fact, we show that a richer class of Boolean functions can be computed by $\QAC^0$.

\begin{theorem}[See also \Cref{cor:sym_ac0_in_qac0}]\label{thm:symac0}
Every symmetric function in $\AC^0$ can be exactly computed in $\QAC^0$.
\end{theorem}

Recall that the exact preparation of the $\ket W$ state requires the $\EX_1$ function to be computed in $\QAC^0$. Since $\EX_1$ is a symmetric function in $\AC^0$, this follows from \Cref{thm:symac0}.
We emphasize that, prior to our work, it was even unknown whether $\EX_1$ can be approximately computed in $\QAC^0$; and \Cref{thm:symac0} answers an open question of Rosenthal \cite{rosenthal_thesis} in the affirmative.

The proof of \Cref{thm:symac0} relies on a classical result of Håstad, Wegner, Wurm, and Yi \cite{symac0} for computing symmetric $\AC^0$ functions with very few wires. 
We observe that their construction can be implemented with $\polylog(n)$-size fanout (which follows from the shrunk version of \Cref{cor:qac0_exact_intro}) and $\polylog(n)$-bit symmetric functions (which follows from the $\QNC^0_\mathrm{wf}$ constructions of \cite{hs_fanout, tt_fanout}). See \Cref{sec:symac0} for details.

\subsection{Future Directions}

We show that having polynomially many classical copies of the input enables $\QAC^0$ to compute highly non-trivial Boolean functions beyond $\AC^0[p]$.
It is reasonable to ask more fine-grained questions about the actual number of copies needed.
Indeed, if this ``copy complexity'' can be reduced to $\polylog(n)$, then $\QAC^0 = \QAC^0_\mathrm{wf}$. 
Towards this question, we make some partial progress in \Cref{sec:better_copy_complexity}. In particular, we show that $\QAC^0$ circuits can compute parity with $O(n^{3/2})$ copies copies of the input and more generally any symmetric function with $O(n^2)$ copies.

Regarding the separation we achieve in \Cref{thm:intro_sep_thm}, it is natural to ask if the soundness be reduced from inverse exponential to zero. This would lead to a separation between \emph{exact} $\QAC^0$ and $\AC^0[p]$. Such a class of decision problems has previously been called $\EQAC^0$ \cite{qacc}. Similarly, another natural question is whether we can exhibit a total function in $\QAC^0$ that is hard against $\AC^0[p]$ circuits on average. 
Note that by the blocky nature of $\Copy\EX_{n/2}$, a random input string is a no instance with high probability; hence this function has high correlation with the constant zero function.  

The techniques introduced in this work do not seem to directly address the question of whether $\AC^0 \subset \QAC^0$. One may view \Cref{thm:symac0} as some partial progress; however, simple $\AC^0$ functions which require large fanout, such as the indexing function, still seem out of reach with the primitives developed thus far.

Another interesting question is to reproduce our result using a finite gate set. Note that the approximate nekomata construction of \cite{grier_morris} only uses Hadamard and Boolean function gates, whereas our approach heavily relies on the exact amplitude amplification, for which arbitrary single-qubit gates seem necessary.

Finally, we review the questions left open in Moore's original work \cite{moore:1999}:
\begin{enumerate}[label=(\arabic*)]
\item Is $\QAC^0 = \QAC^0_\mathrm{wf}$?
\item Is $\QAC^0[p] = \QACC^0$ for arbitrary primes $p>1$?
\item Is $\QAC^0_\mathrm{wf} = \QTC^0$?
\end{enumerate}
For the above questions, Moore remarks:
\begin{center}
\noindent\emph{We conjecture that the answer to all these questions is `no,' but quantum circuits can be surprising}
\end{center}
In \cite{qacc}, it was shown that the answer to (2) is in fact ``yes''. 
Our \Cref{cor:qtc0_exact_intro}, which builds on \cite{grier_morris}, shows that (3) is also true. 
While (1) is still open, we believe our separation and simulation results (\Cref{thm:intro_sep_thm} and \Cref{thm:poly_copies}) exhibit further ways in which quantum circuits can be surprising, suggesting that perhaps the answer to (1) is also ``yes''.

\section{Preliminaries}\label{sec:prelim}

For every integer $n\ge1$, we use $[n]$ to denote the set $\{1,2,\ldots,n\}$.
For every binary string $x\in\{0,1\}^n$, we use $|x|=x_1+\cdots+x_n$ to denote its Hamming weight.
For binary strings $x,y$ of equal length, we use $x\oplus y$ to denote the binary string of their bit-wise XOR.
We use $1_{\mathcal{E}}$ to denote the indicator function of event $\mathcal{E}$, i.e., $1_{\mathcal{E}}=1$ if $\mathcal{E}$ happens and $1_{\mathcal{E}}=0$ if otherwise.
We use $\mathbb{I}_n$ to denote the identity matrix on $n$ qubits. 

\paragraph{Symmetric Boolean Functions.}
We only refer to many-to-one functions when we use the term ``Boolean function''.
For a Boolean function $f\colon\{0,1\}^n\to\{0,1\}$, we say $f$ is symmetric if $f(x)$ depends only on $|x|$.

For integer $k\ge0$, we use $\EX_k$ to denote the exact threshold function of Hamming weight $k$, defined by $\EX_k(x)=1_{|x|=k}$; and use $\Threshold_{\ge k}$ to denote the threshold function of Hamming weight $k$, defined by $\Threshold_{\ge k}(x)=1_{|x|\ge k}$.

For each integer $m\ge2$, we define $\MOD_m\colon\{0,1\}^n\to\{0,1\}$ by $\MOD_m(x)=1_{|x|\equiv1\pmod m}$.
We also reserve $\Parity=\MOD_2$ for the parity function and sometimes use $\Parity_n$ to highlight the input length is $n$.

\paragraph{Classical Circuit Complexity.}
In a Boolean circuit, every gate evaluates some Boolean function on the input wires and forwards the outcome through its output wires.
For a Boolean circuit, 
\begin{itemize} 
\item \emph{depth} is the maximal length from an input bit to output and \emph{size} is the total number of gates;
\item \emph{fanin} is the maximal number of input wires of each gate, \emph{fanout} is the maximal number of output wires of each gate, and \emph{gate set} is the set of different gates.
\end{itemize}
We use $\neg$ to denote the negation gate; and use $\AND$ (resp., $\OR$) to denote the AND (resp., OR) gate/function. Sometimes we use $\AND_n,\OR_n$ to highlight that the function takes $n$ bit input and it will be clear from the context when we omit it.

We will need the following standard classical circuit classes. Here we give informal description and refer readers to textbooks \cite{arora2009computational,jukna2012boolean} for formal definitions.
\begin{itemize}
\item $\NC^0$ is the set of languages that can be exactly decided by constant-depth constant-fanin Boolean circuits of gate set $\{\neg,\AND,\OR\}$.
\item $\AC^0$ is the set of languages that can be exactly decided by constant-depth polynomial-size Boolean circuits of gate set $\{\neg,\AND,\OR\}$.
\item $\AC^0[m]$ is the set of languages that can be exactly decided by constant-depth polynomial-size Boolean circuits of gate set $\{\neg,\AND,\OR,\MOD_m\}$.
$\ACC^0$ is the union of $\AC^0[m]$ for all $m\ge2$.
\item $\TC^0$ is the set of languages that can be exactly decided by constant-depth polynomial-size Boolean circuits of gate set $\{\neg,(\Threshold_{\ge k})_{k\ge0}\}$.
\end{itemize}
We emphasize that we focus on circuit classes for decision problems and it is known that $\NC^0\subsetneq\AC^0\subsetneq\AC^0[p]\subsetneq\TC^0$ for every constant prime $p$ \cite{ajtai_parity,fss_ac0,hastad_thesis,razborov,smolensky}.

\paragraph{(Multi-Qubit) Quantum Gates.}
The quantum fanout gate $\Fanout_n$ is an $(n+1)$-qubit unitary defined by
$$
\Fanout_n\colon\ket b\ket{x_1,\ldots,x_n}\to\ket b\ket{x_1\oplus b,\ldots,x_n\oplus b}
\quad\text{for every $x_1,\ldots,x_n,b\in\{0,1\}$.}
$$
Every Boolean function $f\colon\{0,1\}^n\to\{0,1\}$ naturally induces an $(n+1)$-qubit unitary $U_f$ by
$$
U_f\colon\ket{x_1,\ldots,x_n}\ket b\to\ket{x_1,\ldots,x_n}\ket{b\oplus f(x_1,\ldots,x_n)}
\quad\text{for every $x_1,\ldots,x_n,b\in\{0,1\}$.}
$$
The generalized Toffoli gates correspond to $U_{\AND_n}$ for all $n\ge1$.
When clear, we sometimes use $\AND,\OR$ to denote $U_\AND,U_\OR$.

\paragraph{Quantum Circuits.}
A quantum circuit $C$ is a product of quantum gates.
\begin{itemize}
\item We say $C$ has \emph{depth} $d$ if $C=M_dM_{d-1}\cdots M_1$ where each $M_i$ is a product of gates operating on disjoint sets of qubits. The \emph{size} of $C$ is the total number of gates in $C$ and qubits that $C$ operates on.
\item Let $S$ be a subset of unitaries. We say $C$ has \emph{gate set} $S$ if every gate of $C$ belongs to $S$.
\end{itemize}

Let $f\colon\{0,1\}^n\to\{0,1\}$.
We say $C$ decides $f$ with $a$ ancillas and $\epsilon$ error if for every $x\in\{0,1\}^n$, the last qubit of $C\ket{x}\ket{0^a}\ket0$ measures, in the computational basis, to $\ket{f(x)}$ with probability at least $1-\epsilon$.
In addition, we say 
\begin{itemize} 
\item it has completeness $c$ if for every $x\in f^{-1}(1)$, the last qubit measures to $\ket{1}$ with probability at least $c$; 
\item and has soundness $s$ if for every $x\in f^{-1}(0)$, the last qubit measures to $\ket{1}$ with probability at most $s$.
\end{itemize}
If $C$ decides $f$ with zero error, then we say $C$ exactly decides $f$ and we can assume without loss of generality $C\ket x\ket{0^a}\ket0=\ket x\ket{0^a}\ket{f(x)}$ by standard uncomputation, which incurs an insignificant constant blowup in the depth and size of the circuit.

\paragraph{$\QAC$ and $\QAC^0$.}
We primarily work with $\QAC$ circuits and the corresponding quantum circuit class $\QAC^0$.
See \cite{nielsen2010quantum,bera2010quantum} for a more comprehensive introduction.

A quantum circuit is a $\QAC$ circuit if it only uses single-qubit gates and generalized Toffoli gates.
We distinguish $\QAC^0$ as the following finer classes $\EQAC^0$ and $\BQAC^0$.
\begin{itemize}
\item $\EQAC^0$ is the set of languages that can be exactly decided by constant-depth polynomial-size $\QAC$ circuits.
\item $\BQAC^0$ is the set of languages that can be decided with error\footnote{This $1/3$ is not essential and can be boosted to arbitrarily small constant in a black-box way. If necessary, we sometimes also give more precise error bound in terms of completeness and soundness.} at most $1/3$ by constant-depth polynomial-size $\QAC$ circuits.
\end{itemize}

We emphasize that we allow \emph{arbitrary} single-qubit gates, instead of a finite number of them. We allow ancillary qubits in our quantum circuits, the number of which will be upper bounded by the circuit size. We also remark that the quantum $\OR$ gate $U_\OR\colon\ket{x}\ket b\to\ket x\ket{b\oplus\OR(x)}$ is constant-depth and constant-size in $\QAC$.

\paragraph{Other Quantum Circuit Classes.}
We will occasionally discuss $\QNC^0,\QNC_\mathrm{wf}^0,\QTC^0$ that we briefly describe here.
Let $C$ be a quantum circuit.
\begin{itemize}
\item We say $C$ is a $\QNC$ circuit if it only uses single- and two-qubit gates. 

$\EQNC^0$ (resp., $\BQNC^0$) correspond to languages that can be decided with zero error (resp., $1/3$ error) by constant-depth $\QNC$ circuits.
\item We say $C$ is a $\QNC_\mathrm{wf}$ circuit if it uses single- and two-qubit gates, as well as fanout gates $(\Fanout_n)_{n\ge1}$.

$\EQNC_\mathrm{wf}^0$ (resp., $\BQNC_\mathrm{wf}^0$) correspond to languages that can be decided with zero error (resp., $1/3$ error) by constant-depth polynomial-size $\QNC_\mathrm{wf}$ circuits.
\item We say $C$ is a $\QTC$ circuit if it uses single- and two-qubit gates, as well as quantum threshold gates $(U_{\Threshold_{\ge k}})_{k\ge0}$.

$\EQTC^0$ (resp., $\BQTC^0$) correspond to languages that can be decided with zero error (resp., $1/3$ error) by constant-depth polynomial-size $\QTC$ circuits.
\end{itemize}

\paragraph{Quantum States.}
Let $\ket\psi$ be an $n$-qubit quantum state. We say quantum circuit $C$ prepares $\ket\psi$ with $a$ ancillas if $C\ket{0^n}\ket{0^a}=\ket\psi\ket{0^a}$.
In addition, we say $C$ prepares an $\epsilon$-approximation of $\ket\psi$ if $\|C\ket{0^n}\ket{0^a}-\ket\psi\ket{0^a}\|_2\le\epsilon$.

Let $\ket\phi$ be an $(n+m)$-qubit quantum state. We say $\ket\phi$ is an $n$-qubit \emph{nekomata} if, up to changing ordering of the qubits, $\ket\phi$ equals $\frac1{\sqrt2}\cdot\ket{0^n}\ket{\phi_0}+\frac1{\sqrt2}\ket{1^n}\ket{\phi_1}$ for some normalized $m$-qubit states $\ket{\phi_0},\ket{\phi_1}$.
In \cite{qacc} and \cite{rosenthal}, it was shown that the tasks of computing parity function $\Parity_n$ and preparing fanout unitary $\Fanout_n$ by constant-depth $\QAC$ circuits is equivalent to constructing constant-depth $\QAC$ circuits preparing $n$-qubit nekomata.

\paragraph{(Non-)Uniformity.}
All our circuit upper bounds (i.e., constructions) are \emph{uniform} circuits that can be efficiently extracted from our proofs and descriptions. 
All our circuit lower bounds hold with respect to \emph{non-uniform} circuits.
The (non-)uniformity is not our focus and we do not discuss it in detail.

\section{Exact Amplitude Amplification}\label{sec:exact_aa}

In this section we highlight how previously known amplitude amplification techniques can be adapted in the setting of constant-depth circuits to improve the fanout constructions in \cite{grier_morris} and \cite{rosenthal}. 
In particular, these techniques enable us to \emph{completely remove} any approximation error from previously known constructions. 

This method of amplitude amplification is implicit in Grover's original work \cite{grover1998quantum} (see also \cite[Theorem 2]{brassard2002quantum}) and has been similarly used in recent quantum state preparation works \cite{rosenthal2024efficient,t_count}.
For completeness we include a proof in \Cref{app:sec:exact-aa}.

\begin{restatable}[\cite{grover1998quantum,brassard2002quantum}]{theorem}{thmflaggrover}\label{thm:flag_grover}
Assume $V$ is a depth-$d$ $s$-size $\QAC$ circuit satisfying
$$
V\ket{0^n}\ket0=\sin\theta\ket{\psi_0}\ket0+\cos\theta\ket{\psi_1}\ket1
\quad\text{for $\theta=\frac\pi{4k+2}$ and $k\in\mathbb N$.}
$$
Then there exists a depth-$O(dk)$ $O(sk)$-size $\QAC$ circuit $C$ such that $C\ket{0^n}\ket0 = \ket{\psi_0}\ket0$.
\end{restatable}

We use \Cref{thm:flag_grover} to remove the approximation error in the nekomata constructions of \cite{rosenthal} and \cite{grier_morris}. These consequences are made precise in the following corollaries, from which we immediately derive \Cref{cor:qac0_exact_intro} and \Cref{cor:qtc0_exact_intro}.

\begin{corollary}\label{cor:exact_qac0}
    There exists a constant-depth $O(2^n)$-size $\QAC$ circuit which exactly prepares a nekomata.
\end{corollary}
\begin{corollary}\label{cor:exact_qtc0}
    There exists a constant-depth polynomial-size $\QTC$ circuit which exactly prepares a nekomata. As a consequence, $\BQTC^0 = \BQNC^0_\mathrm{wf}$ and $\EQTC^0 = \EQNC^0_\mathrm{wf}$.
\end{corollary}

\begin{proof}[Proof of \Cref{cor:exact_qac0,cor:exact_qtc0}]
     Suppose our circuit produces an $\epsilon$-approximate nekomata $\ket\psi$ on the first $n$ qubits, where $\epsilon=O(1/\sqrt n)$ in \cite{rosenthal,grier_morris}. We write
     \begin{align*}
         \ket{\psi} = a\ket{0^n}\ket{\psi_0} + b\ket{1^n}\ket{\psi_1} + c\ket{\omega},
     \end{align*}
      where $a^2+b^2+c^2=1$ and $\ket{\omega}$ is not supported on basis states of the form $\ket{x}\ket{y}$ for $x\in \{0^n, 1^n\}$ and arbitrary $y$. 
      
      If we now add two ancilla qubits to our state, we can apply some two-qubit unitary $Q$ on them and obtain a state with weights on the $\ket{0^n},\ket{1^n}$ components which are just right for \Cref{thm:flag_grover}. 
      Precisely, take a two-qubit unitary $Q$ which satisfies 
      \begin{align*}
          Q\ket{00} = p\ket{00} + q\ket{11} + \sqrt{1 - p^2 - q^2}\ket{01},
      \end{align*}
      where $p^2+q^2\le1$ will be set later.
      It is unimportant how $Q$ acts on the other two-qubit computational basis states. Applying $Q$ to the new ancilla qubits, the result is
      \begin{align*}
          \ket{\phi'} = ap\ket{00}\ket{0^n}\ket{\psi_0} + bq\ket{11}\ket{1^n}\ket{\psi_1} + c'\ket{\omega'}
      \end{align*}
      where $a^2p^2 + b^2q^2 + {c'}^2 = 1$ and $\ket{\omega'}$ is unsupported on $\ket{x}\ket{y}$ for $x \in \left\{0^{n+2}, 1^{n+2}\right\}$ and arbitrary $y$. 
      
      Since $\epsilon=O(1/\sqrt n)$, for $n$ sufficiently large we have $a^2, b^2 \in (1/3, 2/3)$ and hence we can pick $p, q \in (-1, 1)$ such that $ap = bq=\frac1{\sqrt2}\sin\frac{\pi}{10}\approx .219$ and $p^2 + q^2 \leq 1$. 
      
      Finally define $f:\{0, 1\}^{n+2} \to \{0, 1\}$ as 
      \begin{align*}
          f(x) = \begin{cases}
              0 & \text{ if } x \in \left\{0^{n+2}, 1^{n+2}\right\},\\
              1 & \text{ otherwise,}
          \end{cases}
      \end{align*}
      and observe that $U_f$, which maps $\ket{u}\ket{v}$ to $\ket{u}\ket{v\oplus f(u)}$ for $u\in\{0,1\}^{n+2}$ and $v\in\{0,1\}$, is in $\QAC^0$.
      
      Combining the above circuits with $U_f$ and a flag qubit, we obtain a constant-depth circuit $V$ for preparing
      \[
      \ket{\phi''} = \sin\frac\pi{10}\left(\frac1{\sqrt2}\ket{0^{n+2}}\ket{\psi_0}+\frac1{\sqrt2}\ket{1^{n+2}}\ket{\psi_1}\right)\ket0+\cos\frac\pi{10}\ket{\omega''}\ket1
      \]
      for some arbitrary state $\ket{\omega''}$ which is again unsupported on $\sket{0^{n + 2}}$ and $\sket{1^{n + 2}}$.
      Then by \Cref{thm:flag_grover}, we obtain an exact $(n+2)$-qubit nekomata $\frac1{\sqrt2}\ket{0^{n+2}}\ket{\psi_0}+\frac1{\sqrt2}\ket{1^{n+2}}\ket{\psi_1}$. This concludes the proof of both corollaries.
\end{proof}

A simple consequence of \Cref{cor:exact_qac0} is that we may shrink the scale into a $\poly(n)$-size $\QAC$ circuit which exactly prepares a $\log(n)$-qubit nekomata and hence $\Fanout_{\log(n)}$.
By iterating this, we can implement fanout of polylogarithmic size.

\begin{corollary}\label{cor:poly-log-fan}
For any constant $d\ge1$, there exists a constant-depth $\poly(n)$-size $\QAC$ circuit which exactly computes $\Fanout_{\log^d(n)}$.
\end{corollary}
\begin{proof}
Abbreviate $k=\log(n)$ and we proceed by induction on $d$.
The base case $d=1$ is precisely \Cref{cor:exact_qac0}.
For $d\ge2$, observe that $\Fanout_{k^d}$ can be achieved by first making a fanout of size $k^{d/2}$ then applying $k^{d/2}-1$ many $\Fanout_{k^{d/2}}$ in parallel.
Since $d$ is constant, this has constant depth and polynomial size, which completes the proof.
\end{proof}

\begin{remark}\label{rmk:PRU-in-QAC0}
In \cite{foxman2025random}, it is shown that approximate pseudorandom unitaries and $t$-designs (for constant $t$) can be constructed in $\QAC^0$. They rely on the approximate nekomata constructions of \cite{rosenthal, grier_morris} and mention that if the error could be removed from these nekomata constructions, then the parameters of their random unitary constructions could be improved. We expect our \Cref{cor:exact_qac0} and \Cref{cor:exact_qtc0} to be helpful for their purposes.
\end{remark}

\subsection{Threshold Functions with Polylogarithmic Weight}

Now we show that \Cref{cor:poly-log-fan} allows us to compute ``small'' threshold functions. For each $k$, define  $\Threshold_{\ge k}\colon\{0,1\}^n\to\{0,1\}$ to be the threshold function with weight $k$, which outputs $1$ iff the input string has Hamming weight at least $k$.
In \Cref{lem:polylog_threshold_in_qac0}, we will show $\Threshold_{\ge k}\in\QAC^0$ for all $k\le\polylog(n)$. Note that $\OR_n$ is equivalent to $\Threshold_{\geq 1}$, so it can trivially be computed in $\QAC^0$.
We remark that until our work, it was not even clear whether $\Threshold_{\ge2}$ is in $\QAC^0$.

\begin{lemma}\label{lem:polylog_threshold_in_qac0}
If $k\le\polylog(n)$, then $\Threshold_{\ge k}$ is in $\EQAC^0$.
\end{lemma}
\begin{proof}
The overall circuit construction is implicit in \cite{symac0,ragde_wigderson}, with certain classical components replaced by their (non-trivial) quantum implementations \cite{hs_fanout,tt_fanout}.

We will need the following number theoretic fact due to \cite{symac0}. Its proof is provided in \Cref{app:sec:exact-aa} for completeness.

\begin{restatable}[{\cite[Lemma 1]{symac0}}]{fact}{fctmoddethash}\label{fct:mod_det_hash}
Let $S\subseteq [n]$. There exists some integer $|S|\le m \leq O\left(|S|^2\log(n)\right)$ such that $i\not\equiv j\pmod m$ for all distinct $i, j \in S$.
\end{restatable}

For each $m\ge k$ and $\ell=0,1,\ldots,m-1$, define
$$
A^m_\ell=\{i\in[n]\ \mid\ i\equiv\ell\pmod m\}.
$$
Now observe that $\Threshold_{\ge k}(x)=1$ iff there exists $S\subseteq[n]$ of size $k$ such that $x_i=1$ for all $i\in S$.
By \Cref{fct:mod_det_hash}, there exists some $k\le m\le L=O(k^2\log(n))$ such that each $i\in S$ lies in $A_\ell^m$ for a different $\ell=0,1,\ldots,m-1$.
This means
\begin{equation}\label{eq:lem:polylog_ex_in_ac0_polylog_fanout_1}
\Threshold_{\ge k}(x)
\le\bigvee_{m=k}^L
\Threshold_{\ge k}\left(\OR(x\mid A_0^m),\OR(x\mid A_1^m),\ldots,\OR(x\mid A_{m-1}^m)\right),
\end{equation}
where $\OR(x\mid S)$ applies the $\OR$ function on $x$'s bits in $S\subseteq[n]$.

On the other hand, notice that $\{A^m_\ell\}_\ell$ partitions $[n]$ into $m$ disjoint sets.
Hence the RHS of \Cref{eq:lem:polylog_ex_in_ac0_polylog_fanout_1} is also a lower bound of the LHS of \Cref{eq:lem:polylog_ex_in_ac0_polylog_fanout_1}.
Hence
\begin{equation}\label{eq:lem:polylog_ex_in_ac0_polylog_fanout_2}
\Threshold_{\ge k}(x)=\bigvee_{m=k}^L
\Threshold_{\ge k}\left(\OR(x\mid A_0^m),\OR(x\mid A_1^m),\ldots,\OR(x\mid A_{m-1}^m)\right).
\end{equation}

Now we convert \Cref{eq:lem:polylog_ex_in_ac0_polylog_fanout_2} into the following desired $\QAC$ circuit.
\begin{itemize}
\item We first make $L=\polylog(n)$ copies of $x$ using $\Fanout_L$ fanout in parallel for each bit of $x$.
By \Cref{cor:poly-log-fan}, this is constant depth and polynomial size.
\item Then for each $k\le m\le L$ and $\ell=0,1,\ldots,m-1$, we compute $y^m_\ell=\OR(x\mid A_\ell^m)$ in parallel separately using the copies in the previous step. This uses a layer of $\OR$ gates.
\item Now for each $k\le m\le L$, we compute $z_m=\Threshold_{\ge k}(y_0^m,y_1^m,\ldots,y_{m-1}^m)$ in parallel separately. This relies on the following fact due to \cite{hs_fanout,tt_fanout}. For completeness we also provide a self-contained proof in \Cref{app:sec:exact-aa}.

\begin{restatable}[\cite{hs_fanout,tt_fanout}]{fact}{fctpolylogsyminacpolylog}\label{fct:polylog_sym_in_ac0_polylog}
Let $f\colon\{0,1\}^m\to\{0,1\}$ be symmetric and $m\le\polylog(n)$.
Then $f$ can be computed exactly by a constant-depth $\poly(n)$-size $\QAC$ circuit.
\end{restatable}

\item Finally we obtain $\Threshold_{\ge k}(x)=\OR(z_k,z_{k+1},\ldots,z_L)$ with an $\OR$ gate.
\qedhere
\end{itemize}
\end{proof}

\subsection{Symmetric Functions in \texorpdfstring{$\AC^0$}{AC0}}\label{sec:symac0}

The small threshold functions from \Cref{lem:polylog_threshold_in_qac0} actually form a complete basis for symmetric functions in $\AC^0$ \cite{moran_symac0,BW_symac0}. This allows us to put $\Sym\AC^0$ inside $\QAC^0$, where $\Sym\AC^0$ is the class of symmetric functions computable in $\AC^0$.

The following known characterization of $\Sym\AC^0$ was proved independently by Moran \cite{moran_symac0} and Brustman and Wegner \cite{BW_symac0}.

\begin{theorem}[\cite{moran_symac0,BW_symac0}]\label{thm:symAC0characterization}
Let $f\colon\{0,1\}^n\to\{0,1\}$ be symmetric such that $f(x)=v_k$ for all $k=0,1,\ldots,n$ and $x\in\{0,1\}^n$ with Hamming weight $k$. 
Then $f\in\AC^0$ iff $v_k=v_{k+1}=\cdots=v_{n-k}$ for some $k=\polylog(n)$.
\end{theorem}

Now we prove \Cref{thm:symac0} as the following \Cref{cor:sym_ac0_in_qac0}.

\begin{corollary}\label{cor:sym_ac0_in_qac0}
$\Sym\AC^0 \subseteq \EQAC^0$.
\end{corollary}
\begin{proof}
We first observe that $\Threshold_{\ge n-k}(x)=\neg\Threshold_{\ge k+1}(\neg x)$, where $\neg x$ is the bitwise negation of $x$.
Hence by \Cref{lem:polylog_threshold_in_qac0}, $\EQAC^0$ computes $\Threshold_{\ge k}$ for all $k\le\polylog(n)$ and $k\ge n-\polylog(n)$.

For each $k$, define $\EX_k\colon\{0,1\}^n\to\{0,1\}$ as the exact Hamming weight function with weight $k$, which outputs $1$ iff the input string has Hamming weight exactly $k$. Since $\EX_k(x)=\Threshold_{\ge k}(x)\land\neg\Threshold_{\ge k+1}(x)$, we know that $\EX_k\in\QAC^0$ for all $k\le\polylog(n)$ and $k\ge n-\polylog(n)$.

By \Cref{thm:symAC0characterization}, $f\in\Sym\AC^0$ iff it is an $\OR$ (or $\neg\OR$) of $\polylog(n)$ many $\EX_k$ where $k\le\polylog(n)$ or $k\ge n-\polylog(n)$.
Hence it suffices to make $\polylog(n)$ copies of the input string by \Cref{cor:poly-log-fan}, then compute the corresponding $\EX_k$'s in parallel, and finally use $\OR$ or $\neg\OR$ to aggregate the values.
\end{proof}

\section{State-Unitary Duality in \texorpdfstring{$\QAC^0$}{QAC0}}

In this section, we explore the duality between certain states and unitaries.
Such a duality was observed by Moore \cite{moore:1999} between the cat state and the parity function. In particular, one can implement the unitary for parity using a circuit which prepares the cat state.

Here we exhibit another duality, which allows us to \emph{weakly} compute the exact Hamming weight functions using circuits preparing the $\ket W$ state:
$$
\ket W=\frac1{\sqrt n}\sum_{i\in[n]}\ket{e_i}
\quad\text{where $e_i=0^{i-1}10^{n-i}$.}
$$
Then we lift this weak computation to separate $\QAC^0$ from classical circuit classes.

\subsection{The \texorpdfstring{$\ket W$}{W} State and Hamming Weight Tests}\label{sec:construct_w}

We begin with a $\QAC^0$ circuit constructing the $\ket W$ state.
The idea is to use a simple product state to approximate $\ket W$, then use amplitude amplification to obtain $\ket W$ exactly.
This is formalized as the following \Cref{thm:exact_w_1}, which is exactly \Cref{thm:exact_w}.

\begin{theorem}\label{thm:exact_w_1}
There exists a constant-depth polynomial-size $\QAC$ circuit $U$ such that $U\ket{0^n}\ket{0^a}=\ket W\ket{0^a}$ where $a=\poly(n)$.
\end{theorem}
\begin{proof}
Let $\delta\in[0,1]$ be a parameter to be determined later.
Using a layer of single-qubit gates, we prepare 
\begin{align*}
\left(\sqrt\delta\ket0+\sqrt{1-\delta}\ket1\right)^{\otimes n}
=\sqrt{n\delta^{n-1}(1-\delta)}\cdot\ket W+\sket{W^\bot},
\end{align*}
where $\ket{W^\bot}$ is some unnormalized state only supported on computational basis states of Hamming weight not equal to $1$.
Then by the $\EQAC^0$ implementation of $\EX_1$ from \Cref{cor:sym_ac0_in_qac0}, we obtain the state
\begin{equation}\label{eq:thm:exact_w_1}
\sqrt{n\delta^{n-1}(1-\delta)}\cdot\ket W\ket0+\sket{W^\bot}\ket1.
\end{equation}
Observe that $h(\delta)=\sqrt{n\delta^{n-1}(1-\delta)}$ satisfies $h(0)=0$ and
$$
h\left(1-\frac1n\right)=\left(1-\frac1n\right)^{(n-1)/2}\ge\frac1{\sqrt e}\ge\frac12.
$$
Hence by continuity, we can pick $\delta$ such that $h(\delta)=1/2$ and \Cref{eq:thm:exact_w_1} becomes
$$
\frac12\ket W\ket0+\sket{W^\bot}\ket1=\sin\frac\pi6\ket W\ket0+\sket{W^\bot}\ket1.
$$
By \Cref{thm:flag_grover} with the above state preparation procedure, we obtain the desired circuit $U$ for exactly preparing $\ket W$.
\end{proof}

We now show how to weakly compute any fixed Hamming weight using $\ket W$ above.
We start with Hamming weight exactly $n/2$.

\begin{lemma}\label{lem:w_test}
There exist a constant-depth polynomial-size $\QAC$ circuit $C$ such that for every $x \in \{0, 1\}^n$ and $b\in\{0,1\}$, we have $C\ket{x,0^{n+a},b}=\ket x\ket{\psi_{x,b}}$ and
$$
\ket{\psi_{x,b}}=\frac{n-2|x|}n\ket{0^{n+a}}\ket{b}+\sqrt{1-\left(\frac{n-2|x|}n\right)^2}\ket{\phi_x}\ket{b\oplus1},
$$
where $a=\poly(n)$ and $\ket{\phi_x}$ is a normalized state orthogonal to $\ket{0^{n+a}}$ and depending on $x$.
\end{lemma}
\begin{proof}
Let $U$ be the circuit which prepares $\ket W$ from \Cref{thm:exact_w} using $a$ ancilla qubits.
Consider the circuit $C$ in \Cref{fig:w-test}.
Observe that after the first two layers of $C$, the state is
\begin{align*}
 \ket{x}\ket{0^a}\ket{W_x}\ket b
\quad\text{where }\ket{W_x}=\frac{1}{\sqrt{n}}\sum_{i = 1}^n (-1)^{x_i}\ket{e_i}.
\end{align*}
Note that $\braket{W}{W_x}=\frac{n-2|x|}n$. Hence we express
\begin{equation}\label{eq:lem:w_test_1}
\ket{W_x}=\frac{n-2|x|}n\ket W+\sqrt{1-\left(\frac{n-2|x|}n\right)^2}\sket{W_x^\bot},
\end{equation}
where $\sket{W_x^\bot}$ is a normalized state depending on $x$ and is orthogonal to $\ket W$.
We now analyze the evolution of $\ket x\ket{0^a}\ket W\ket b$ and $\ket x\ket{0^a}\sket{W_x^\bot}\ket b$ separately.

For $\ket x\ket{0^a}\ket W\ket b$, the subsequent layers of $C$ have the following effect:
\begin{align*}
\ket x \ket{0^a} \ket W\ket b 
\xrightarrow{U^\dag}\ket x \ket{0^a} \ket{0^n} \ket b 
\xrightarrow{\OR}\ket x \ket{0^a} \ket{0^n} \ket b.
\end{align*}

For $\ket x\ket{0^a}\sket{W_x^\bot}\ket b$, since $\ket{0^a}\ket W$ is orthogonal to $\ket{0^a}\ket{W^\bot}$, $\ket{\phi_x}:=U^\dag\ket{0^a}\sket{W_x^\bot}$ is also orthogonal to $U^\dag\ket{0^a}\ket{W}=\ket{0^{n+a}}$. Hence the $\OR$ operation always flips $\ket b$, giving the overall state $\ket x\ket{\phi_x}\ket{b\oplus1}$.

By linearity and \Cref{eq:lem:w_test_1}, we combine the above two evolutions and obtain
$$
\ket{\psi_{x,b}}=\frac{n-2|x|}n\ket{0^{n+a}}\ket{b}+\sqrt{1-\left(\frac{n-2|x|}n\right)^2}\ket{\phi_x}\ket{b\oplus1}
$$
as claimed.
\end{proof}

\Cref{lem:w_test} shows how to weakly decide $\EX_{n/2}$: after measuring the final register (initialized as $b=0$), we always output $1$ if $\EX_{n/2}(x)=1$; and output $0$ with probability at least $1/n^2$ if $\EX_{n/2}(x)=0$.
By padding with sufficiently many $1$'s or $0$'s, we can modify this construction to detect any other fixed Hamming weight in a similar manner.

\begin{corollary}\label{cor:general_w_test}
For every integer $0\le k\le n$, the function $\EX_k$ can be weakly decided in $\QAC^0$. 

That is, there is a constant-depth polynomial-size $\QAC$ circuit that always outputs $1$ if $\EX_k(x)=1$; and outputs $0$ with probability at least $1/n^2$ if $\EX_k(x)=0$.
\end{corollary}
\begin{proof}
Given an $n$-bit input $x$, we pad it with an $n$-bit fixed string $1^{n-k}0^k$ as $y=x\circ1^{n-k}0^k$.
Let $m=2n$.
Then $\EX_k(x)=\EX_{m/2}(y)$, for which we can use \Cref{lem:w_test}.
\end{proof}

\Cref{cor:general_w_test} allows $\QAC^0$ to compute arbitrary symmetric function, assuming we have classical copies of the input strings. The consequence of this is discussed in the following \Cref{sec:sep_qac0_lightweight}.

\subsection{Separating \texorpdfstring{$\QAC^0$}{QAC0} from Classical Classes}\label{sec:sep_qac0_lightweight}

The one-sided error in \Cref{cor:general_w_test} is a very important feature that allows us to boost the success probability using $\AND$ instead of majority. This is crucial as it is unclear whether majority is in $\QAC^0$, but $\AND$ can be used by definition.

\begin{theorem}\label{thm:qac0_copy_sym}
Let $f\colon\{0,1\}^n\to\{0,1\}$ be an arbitrary symmetric function.
For every integer $r\ge1$, there exists a constant-depth polynomial-size $\QAC$ circuit such that, given $m=r\cdot(n+1)^3$ identical copies of $n$-bit string $x$, it computes $f(x)$ with completeness $1$ and soundness $2^{-\Omega(r)}$.

That is, the circuit, on input $\sket{\underbrace{x,x,\ldots,x}_{m\text{ copies}}}$, always outputs $1$ if $f(x)=1$; and outputs $0$ with probability at least $1-2^{-\Omega(r)}$ if $f(x)=0$.
\end{theorem}
\begin{proof}
Since $f$ is symmetric, we can express it as $f(x)=\bigvee_{k\in S}\EX_k(x)$ for some $S\subseteq\{0,1,\ldots,n\}$.
For each $\EX_k(x)$, we take the $\AND$ of $r\cdot n^2$ independent trials of \Cref{cor:general_w_test}. This requires $r\cdot n^2$ copies of $x$. Moreover, it has completeness $1$ and soundness $(1-n^{-2})^{r\cdot n^2}=2^{-\Omega(r)}$ to compute $\EX_k(x)$.
Then we take the $\OR$ of the above values to compute $f(x)$, which gives completeness $1$ and soundness $r\cdot2^{-\Omega(r)}=2^{-\Omega(r)}$.
The total number of copies of $x$ we need is $|S|\cdot r\cdot n^2\le r\cdot(n+1)^3$.
\end{proof}

Given classical copies of the input, \Cref{thm:qac0_copy_sym} allows us to replace any symmetric gate in a classical circuit with a $\QAC$ circuit of polynomial size. This classical class is exactly $\TC^0$ and the argument is similar to the standard way of converting circuits to formulas. The following \Cref{thm:tc0_in_qac0_copy} formalizes this and proves \Cref{thm:poly_copies}.

\begin{theorem}\label{thm:tc0_in_qac0_copy}
Assume $f\colon\{0,1\}^n\to\{0,1\}$ is exactly computed by a constant-depth polynomial-size classical circuit of symmetric gates.
Then there exists a constant-depth polynomial-size $\QAC$ circuit such that, given $m=\poly(n)$ identical copies of $n$-bit string $x$, it computes $f(x)$ with error $2^{-\poly(n)}$.

That is, the circuit, taking input $\sket{\underbrace{x,x,\ldots,x}_{m\text{ copies}}}$, outputs $f(x)$ with probability at least $1-2^{-\poly(n)}$.
\end{theorem}
\begin{proof}
By \Cref{thm:qac0_copy_sym}, each symmetric gate of $f$ can be computed by a constant-depth polynomial-size $\QAC$ circuit with error $2^{-\poly(n)}$, given $\poly(n)$ identical copies of its fan-in values.
Hence we can convert the classical circuit of $f$ in a top-down fashion: each time we replace the current symmetric gate by a small $\QAC$ circuit with error $2^{-\poly(n)}$ that demands $\poly(n)$ identical copies of its fan-in values; then we recursively expand the copies of its fan-in gates in a same way.
Since $f$ has constant depth and polynomial size, in the end we just need $\poly(n)$ copies of the input string. In addition, the error is $\poly(n)\cdot2^{-\poly(n)}=2^{-\poly(n)}$ by union bound.
\end{proof}

\Cref{thm:tc0_in_qac0_copy} already allows us to obtain \emph{partial} functions that separate $\QAC^0$ from $\AC^0$ or even $\AC^0[p]$. This is done by choosing $f$ to be the majority function and defining a lifted version that takes polynomially many identical copies of the input string. Fortunately, \emph{checking} that a given string has this ``blocky'' form can be easily done with generalized Toffoli gates. This observation leads to a \emph{total} function which witnesses a separation. We make this formal below.

Let $f\colon\{0,1\}^n\to\{0,1\}$.
For each integer $m\ge1$, define $f^{\uparrow m}\colon(\{0,1\}^n)^m\to\{0,1\}$ by
\begin{equation}\label{eq:f_lift}
f^{\uparrow k}(x^{(1)},\ldots,x^{(m)})=\begin{cases}
f(x^{(1)}) & x^{(1)}=\cdots=x^{(m)},\\
0 & \text{otherwise.}
\end{cases}
\end{equation}
Note that this is the same as the lifting gadget $\Copy$ introduced in \Cref{sec:overview}. We use $\uparrow k$ here to highlight the number of input copies.

\begin{lemma}\label{lem:copy_lift}
Let $f\colon\{0,1\}^n\to\{0,1\}$ be a Boolean function that can be computed by a constant-depth polynomial-size $\QAC$ circuit with $m\le\poly(n)$ identical input copies and error $\epsilon$.
Then $f^{\uparrow m}\colon(\{0,1\}^n)^m\to\{0,1\}$ defined in \Cref{eq:f_lift} can be computed by a constant-depth polynomial-size $\QAC$ circuit with error $\epsilon$.
\end{lemma}
\begin{proof}
We identify the $nm$-bit input as $\ket{x^{(1)}}\cdots\ket{x^{(m)}}$ where each $x^{(i)}$ has $n$ bits.
We can check the equality of $x^{(1)},\ldots,x^{(m)}$ by checking $x^{(1)}_j=\cdots=x^{(m)}_j$ for every coordinate $j\in[n]$. This is achieved by computing $\EX_0\lor\EX_m$ in parallel for each coordinate, then combining the coordinate-checks with $\AND$, which has constant depth and size $\poly(m)=\poly(n)$.

In parallel with the above check, we assume the inputs are equal and use the $\QAC$ circuit (given $m$ identical input copies) of $f$ to compute its value.

The final outcome is an $\AND$ of the two computations above. Since the first equality check does not make error, the error only comes from the second part which is exactly $\epsilon$ by assumption.
\end{proof}

Then we have the following corollary.

\begin{corollary}\label{cor:tc_0_copy_in_qac0}
Assume $f\colon\{0,1\}^n\to\{0,1\}$ is exactly computed by a constant-depth polynomial-size classical circuit of symmetric gates.
Then there exists some $m\le\poly(n)$ such that $f^{\uparrow m}$ can be computed by a constant-depth polynomial-size $\QAC$ circuit with error $2^{-\poly(n)}$.
\end{corollary}
\begin{proof}
We simply combine \Cref{thm:tc0_in_qac0_copy} and \Cref{lem:copy_lift}.
\end{proof}
By known separations in classical complexity theory, we obtain the following theorem that separates $\QAC^0$ from $\AC^0[p]\supsetneq\AC^0$, which proves \Cref{thm:intro_sep_thm}.

\begin{theorem}\label{thm:qac0_sep}
$\BQAC^0 \not \subset \AC^0[p]$.
Moreover if $\ACC^0 \not = \TC^0$, then $\BQAC^0 \not \subset \ACC^0$.
\end{theorem}
\begin{proof}
Recall that the majority function $\Maj$ is not contained in $\AC^0[p]$ \cite{razborov,smolensky}.
Hence $\Maj^{\uparrow\poly(n)}\notin\AC^0[p]$ since $\AC^0[p]$ can make copies of the input string for free.
On the other hand by \Cref{cor:tc_0_copy_in_qac0} and since $\Maj\in\TC^0$, we know $\Maj^{\uparrow\poly(n)}\in\BQAC^0$. This shows $\BQAC^0 \not \subset \AC^0[p]$.

The second result also follows from \Cref{cor:tc_0_copy_in_qac0} by taking any language in $\TC^0\setminus\ACC^0$ and noticing that classical copies are free for $\ACC^0$ as well.
\end{proof}

\section*{Acknowledgements}
JM thanks Farzan Byramji, Sabee Grewal, Dale Jacobs, Kunal Marwaha, and Gregory Rosenthal for inspiring discussions.
KW is supported by the National Science Foundation under Grant No. DMS-2424441, and by the IAS School of Mathematics.

\bibliography{bibli}
\bibliographystyle{alphaurl}

\appendix

\section{Missing Proofs From Section \ref{sec:exact_aa}}\label{app:sec:exact-aa}

\thmflaggrover*
\begin{proof}
Let $R_1 = \mathbb{I}_{n}\otimes Z$ and $R_2 =  \mathbb{I}_{n + 1} - 2\proj{\alpha_1}$ be two reflections.
Note that $R_2$ has depth $O(d)$ and size $O(s)$ since
    \[
        R_2=\mathbb{I}_{n + 1} - 2\proj{\alpha_1} = V(\mathbb{I}_{n + 1} - 2\proj{0^{n + 1}})V^{\dag}.
    \]
    For $t\in\mathbb N$, define 
    \[
        \ket{\alpha_t} = \sin{t\theta}\ket{\psi_0}\ket{0} + \cos{t\theta}\sket{\psi_1}\ket{1}
    \]
    and observe that
    \begin{align*}
        R_2R_1\ket{\alpha_t} 
        &= (\mathbb{I}_{n + 1} - 2\proj{\alpha_1}) \left( \sin{t\theta}\ket{\psi_0}\ket{0} - \cos{t\theta}\sket{\psi_1}\ket{1} \right) \\
        &= (-2\sin^2{\theta}\sin{t\theta} +\sin{2\theta}\cos{t\theta} + \sin{t\theta})\ket{\psi_0}\ket{0}\\
        &\qquad\qquad+ (-\sin{2\theta}\sin{t\theta} + 2\cos{t\theta}\cos^2{\theta} - \cos{t\theta})\sket{\psi_1}\ket{1}\\
        &= \sin{((t + 2)\theta)}\ket{\psi_0}\ket{0} + \cos{((t + 2)\theta)}\sket{\psi_1}\ket{1}\\
        &= \ket{\alpha_{t + 2}}.
    \end{align*}
    Define $C=(R_2R_1)^kV$.
    Since $\theta=\frac\pi{4k+2}$, we have $C\ket{0^{n+1}}=\ket{\alpha_{2k + 1}} = \ket{\psi_0}\ket{0}$ as desired.
\end{proof}

\fctmoddethash*
\begin{proof}
Let $m\ge2$ be the smallest integer such that $i\not\equiv j\pmod m$ for all $i\ne j\in S$.
Note that $m$ exists and in particular $m\le n+1$.
In addition, $m\ge|S|$ since different $i\in S$ needs to occupy a different residue modulo $m$.

By the choice of $m$, any integer $m'<m$ divides $|i-j|$ for some $i\ne j\in S$.
This implies that the least common multiple of $1,2,\ldots,m-1$, denoted $\mathsf{lcm}(1,2,\ldots,m-1)$, divides $\prod_{i\ne j\in S}|i-j|$.
In particular,
$$
\mathsf{lcm}(1,2,\ldots,m-1)\le\prod_{i\ne j\in S}|i-j|\le n^{|S|^2}
$$
as $i,j\in S\subseteq[n]$.
On the other hand, it is known that the second Chebyshev function $\psi(m-1):=\ln\mathsf{lcm}(1,2,\ldots,m-1)$ has the asymptotic behavior $\psi(x)=\Theta(x)$ (see e.g., \cite{wiki:chebyfunc}).
Hence
$$
m\le2(m-1)=\Theta(\psi(m-1))\le O\left(|S|^2\log(n)\right)
$$
as desired.
\end{proof}

\fctpolylogsyminacpolylog*

\begin{proof}
Recall that $\EX_k$ is the exact Hamming weight function with weight $k$. We first construct the $\QAC$ circuit for each $\EX_k$.
For the general $f$, observe that $f(x)=\bigvee_{k\in S}\EX_k(x)$ for some $S\subseteq\{0,1,\ldots,m\}$.
Hence, to compute $f(x)$, we first make $|S|\le\polylog(n)$ copies of $x$ using $\Fanout_{|S|}$ from \Cref{cor:poly-log-fan}.
Then we compute $\EX_k(x)$ in parallel separately using those copies and finally obtain $f(x)$ by taking an $\OR$ of the results.

Now we focus on $\EX_k$.
Define $r=\left\lfloor\log(m)\right\rfloor$ and $|x|$ as the Hamming weight of $x$ and
$$
\ket{\phi_t}=\frac{\ket{0}+e^{i\pi(|x|-k)/2^t}\ket1}{\sqrt2}
\quad\text{for each }t=0,1,\ldots,r.
$$
Note that $\ket{\phi_0}\ket{\phi_1}\cdots\ket{\phi_r}$ can be construct in constant depth and polynomial size as follows.
\begin{itemize}
\item We first make $r+1\le\polylog(n)$ copies of $x$ using $\Fanout_{r+1}$ from \Cref{cor:poly-log-fan}.
\item In parallel for each $t$, we use the $t$th copy of $x$ to construct $\ket{\phi_t}$:
\begin{align*}
\ket{x}\ket{0^m}
&\to\ket{x}\frac{\ket0+e^{-i\pi k/2^t}\ket1}{\sqrt2}\ket{0^{m-1}}
\tag{a single-qubit rotation gate}\\
&\to\ket{x}\frac{\ket{0^m}+e^{-i\pi k/2^t}\ket{1^m}}{\sqrt2}
\tag{a $\Fanout_{m-1}$ gate using \Cref{cor:poly-log-fan}}\\
&\to\ket{x}\frac{\ket{0^m}+e^{-i\pi k/2^t}\cdot\prod_{j\in[m]}e^{i\pi x_j/2^t}\ket{1^m}}{\sqrt2}
\tag{a layer of controlled $(\pi/2^t)$-phase gate}\\
&\to\ket{x}\frac{\ket{0}+e^{-i\pi k/2^t}\cdot\prod_{j\in[m]}e^{i\pi x_j/2^t}\ket{1}}{\sqrt2}\ket{0^{n-1}}
\tag{a $\Fanout_{m-1}$ gate using \Cref{cor:poly-log-fan}}\\
&=\ket{x}\ket{\phi_t}\ket{0^{m-1}}.
\end{align*}
Note that $m=\polylog(n)$ and thus the above operations are constant depth and $\poly(n)$ size.
\end{itemize}

By a layer of Hadamard and $X$ gates, we obtain $\ket{\psi_0}\ket{\psi_1}\cdots\ket{\psi_r}$ where $\ket{\psi_t}=XH\ket{\phi_t}$.
Then we apply the generalized Toffoli gate $\AND$ on $\ket{\psi_0}\ket{\psi_1}\cdots\ket{\psi_r}$ and store the value in another register.
Now we prove that this value equals $\EX_k(x)$.
\begin{itemize}
\item If $\EX_k(x)=1$, then $|x|=1$ and $\ket{\phi_t}=\ket+$ for all $t=0,1,\ldots,r$. Hence $\ket{\psi_t}=XH\ket+=\ket1$ for all $t$ and the final $\AND$ value equals $1$ as well.
\item If $\EX_k(x)=0$, then $|x|-k=2^{t^*}v$ for some $t^*=0,1,\ldots,r$ and odd integer $v$.
Therefore $\ket{\phi_{t^*}}=\frac{\ket0+e^{i\pi v}\ket1}{\sqrt2}=\ket-$.
Hence $\ket{\psi_{t^*}}=XH\ket-=\ket0$ and the final $\AND$ value equals $0$.
\end{itemize}

Finally we clean up the workspace by undoing the above gates except the last $\AND$. This gives the desired circuit for $\EX_k$ and also general $f$.
\end{proof}

\section{Indexing Requires Large Classical Fanout}\label{sec:index_fanout}

In this section, we show that juntas and the indexing function requires large fanout, even in $\TC^0$. This is a simple counting argument that we provide for completeness.

We start with juntas of logarithmic size.

\begin{theorem}\label{thm:junta_needs_fanout}
There exists a function $f\colon\{0,1\}^{\log(n)}\to\{0,1\}$ such that, if $f$ can be computed by a constant-depth circuit of gate set $\{\neg,(\Threshold_{\ge k})_{k\ge0})\}$, then the circuit has fanout $n^{\Omega(1)}$.
\end{theorem}
\begin{proof}
Any circuit of depth $d$ and fanout $r$ has at most $O(r^d)$ gates.
Since each gate has $O(r^d)+\log(n)$ possibilities, there are at most $(r^d+\log(n))^{O(r^d)}$ many such circuits.
On the other hand, there are $2^n$ many distinct Boolean functions on $\log(n)$ bits, which requires distinct circuits.
This means $r=n^{\Omega(1)}$ as $d=O(1)$.
\end{proof}

Recall the definition of the indexing function $\Ind \colon \{0, 1\}^n \times \{0,1\}^{\log(n)} \to \{0, 1\}$ that
\begin{align*}
    \Ind(x_1,\dots x_n, i_1, \dots, i_{\log(n)}) = x_{i_1 \cdots i_{\log(n)}}.
\end{align*}

\begin{corollary}\label{cor:indexing_needs_fanout}
If $\Ind$ can be computed by a constant-depth circuit of gate set $\{\neg,(\Threshold_{\ge k})_{k\ge0}\}$, then the circuit has fanout $n^{\Omega(1)}$.
\end{corollary}
\begin{proof}
For every $x\in\{0,1\}^n$, we define $\Ind_x\colon\{0,1\}^{\log(n)}\to\{0,1\}$ by 
$$
\Ind_x(i_1,\ldots,i_{\log(n)})=\Ind(x,i_1,\ldots,i_{\log(n)}).
$$
Then $\Ind_x$ enumerates all Boolean functions on $\log(n)$ bits, including the one from \Cref{thm:junta_needs_fanout}.
This completes the proof.
\end{proof}

\section{Towards Better Copy Complexity}\label{sec:better_copy_complexity}

In this section, we present some additional techniques to improve the copy complexity of various symmetric functions.

\subsection{Truncated Parallel Repetition}

The computation of $\EX_{n/2}(x)$ in \Cref{sec:construct_w} is rather weak in that it may only be correct with probability $1/n^2$ if $|x|\ne n/2$ (see \Cref{cor:general_w_test}).
In this part, we show how to moderately improve this to $\polylog(n)/n$.

In general, this error reduction is achieved by repeating the test in parallel, which is the key idea in \Cref{sec:sep_qac0_lightweight}, i.e., there we simply consume more copies of the input string in order to perform more runs the $W$ test.

However in $\QAC^0$, we only know how to make $\polylog(n)$ copies via \Cref{cor:poly-log-fan}, which translates to $\polylog(n)$ parallel runs and a success rate of $\polylog(n)/n^2$.

We show in \Cref{thm:parallel_rep_trunc} below how to achieve $\polylog(n)/n$ success rate, which is comparable to per $n\cdot\polylog(n)$ parallel runs.

\begin{theorem}\label{thm:parallel_rep_trunc}
For every $n$, there exists a constant-depth polynomial-size $\QAC$ circuit such that 
\begin{itemize} 
\item it outputs $1$ with probability at least $1-2^{-\polylog(n)}$ if $\EX_{n/2}(x)=1$;
\item and outputs $0$ with probability at least $\polylog(n)/n$ if $\EX_{n/2}(x)=0$.
\end{itemize}
\end{theorem}
\begin{proof}
We assume without loss of generality that $n$ is an even number and the overall circuit is depicted in \Cref{fig:parallel_w_test}.

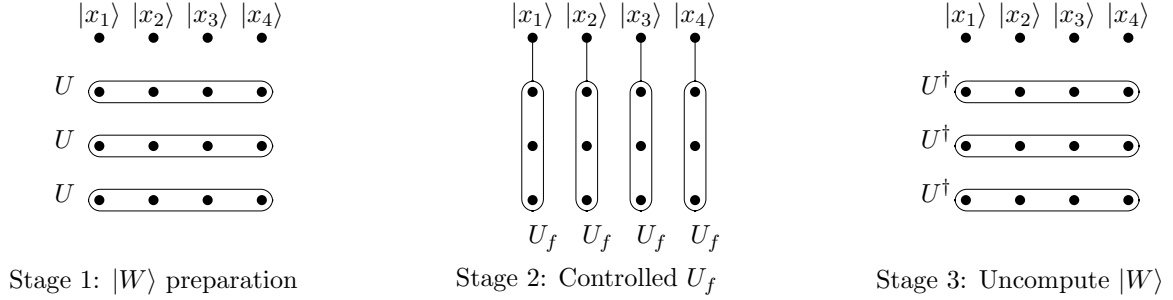
\begin{figure}[H]
\centering
\scalebox{0.9}{\newcommand{\xmax}{4}
\newcommand{\ymax}{4}
\newcommand{\tikzscale}{.8}
\newcommand{\xpicshift}{4}
\newcommand{\stageoffset}{.5}

\begin{tikzpicture}[scale=\tikzscale]
  \begin{scope}[]
  \foreach \x in {1,...,\xmax} {
      \foreach \y in {1,...,\ymax} {
          \node[circle, fill, inner sep=1.5pt] (N\x\y) at (\x,\y) {};
      }
      \node[anchor=south] at (\x,\ymax) {$\ket{x_{\x}}$};
  }
  \foreach \y in {1,..., 3} {
      \draw[rounded corners=5pt] (.8, \y-.2) rectangle (\xmax+0.2, \y+.2);
      \node[anchor=south west] at (0, \y - .2) {$U$};
}
\end{scope}

\begin{scope}[shift={(\xmax + \xpicshift,0)}]

  \foreach \x in {1,...,\xmax} {
    \node[anchor=south] at (\x,\ymax) {$\ket{x_{\x}}$};
      \foreach \y in {1,...,\ymax} {
          \node[circle, fill, inner sep=1.5pt] (N\x\y) at (\x,\y) {};
      }
  }

    \foreach \x in {1,...,\xmax} {
      \draw[rounded corners=5pt] (\x-0.2, .8) rectangle (\x+0.2, \ymax -.8);
      \node[anchor=north] at (\x + .2, .7) {$U_f$};
        }
    \foreach \x in {1,...,\xmax} {
      \draw (\x, \ymax - .8) -- (\x, \ymax);
        }
\end{scope}
\begin{scope}[shift={(2*\xmax + 2*\xpicshift,0)}]
  \foreach \x in {1,...,\xmax} {
      \foreach \y in {1,...,\ymax} {
          \node[circle, fill, inner sep=1.5pt] (N\x\y) at (\x,\y) {};
      }
      \node[anchor=south] at (\x,\ymax) {$\ket{x_{\x}}$};
  }
  \foreach \y in {1,..., 3} {
      \draw[rounded corners=5pt] (.8, \y-.2) rectangle (\xmax+0.2, \y+.2);
      \node[anchor=south west] at (0, \y - .2) {$U^{\dag}$};
}
\end{scope}
\node at (\xmax/2,-\stageoffset) {\shortstack{Stage 1: $\ket{W}$ preparation}};
\node at (\xmax/2+\xmax+\xpicshift,-\stageoffset) {\shortstack{Stage 2: Controlled $U_f$}};
\node at (\xmax/2+2*\xmax+2*\xpicshift,-\stageoffset) {\shortstack{Stage 3: Uncompute $\ket{W}$}};
\end{tikzpicture}}
\caption{The circuit of truncated parallel repetition for $n=4$.}\label{fig:parallel_w_test}
\end{figure}

Let $m=n\cdot\polylog(n)$.
We will show that there exists a constant-depth polynomial-size $\QAC$ circuit that, on $\ket x$, prepares a state is extremely close to
\begin{align}\label{eq:thm:parallel_rep_trunc_1}
\ket{\phi_x} = \ket{x}\otimes \bigg{(}\underbrace{\frac{1}{\sqrt{n}}\sum_{i = 1}^n (-1)^{x_i}\ket{e_i}}_{\ket{W_x}}\bigg{)}^{\otimes m},
\end{align}
given which we can finish the proof with the following \Cref{clm:thm:parallel_rep_trunc_1}.

\begin{claim}\label{clm:thm:parallel_rep_trunc_1}
There exists a constant-depth polynomial-size $\QAC$ circuit such that, on $\ket{\phi_x}$, it always outputs $1$ if $\EX_{n/2}(x)=1$; and outputs $0$ with probability at least $1-(1-1/n^2)^{m}=\polylog(n)/n$ if $\EX_{n/2}(x)=0$.
\end{claim}
\begin{proof}
Let $U$ be the circuit in \Cref{thm:exact_w} preparing $\ket W$.
The proof of \Cref{lem:w_test} shows $\OR_n\circ (U^\dag\otimes\mathbb{I}_1)(\ket{W_x}\ket0)=\frac{n-2|x|}n\ket{0^n}\ket0+\sqrt{1-\left(\frac{n-2|x|}n\right)^2}\ket{\tau_x}\ket 1$ where $\ket{\tau_x}$ is orthogonal to $\ket{0^n}$ and $\OR_n$ puts the $\OR$ outcome of the first $n$ qubits in the $(n+1)$th qubit.
This means
\begin{equation}\label{eq:clm:thm:parallel_rep_trunc_1_1}
(\OR_n\circ(U^\dag\otimes\mathbb{I}_1))^{\otimes m}(\ket{W_x}\ket0)^{\otimes m}
=\left(1-\left(\frac{n-2|x|}n\right)^{2}\right)^{m/2}\ket{\tau_x}^{\otimes m}\ket{1^m}+\sum_{b\in\{0,1\}^m\setminus\{1^m\}}\ket{\star_b}\ket{b},
\end{equation}
where each $\ket{\star_b}$ is an unnormalized $nm$-qubit state.
Now we apply an $\AND$ gate of the last $m$ qubits and store the answer in an additional ancilla.
\begin{itemize}
\item If $\EX_{n/2}(x)=1$, then \Cref{eq:clm:thm:parallel_rep_trunc_1_1} is simply $\ket{\tau_x}^{\otimes m}\ket{1^m}$ and the $\AND$ outcome is a deterministic $1$;
\item otherwise $\EX_{n/2}(x)=0$, then $|x|\ne n/2$ and the amplitude of $\ket{\tau_x}^{\otimes m}\ket{1^m}$ in \Cref{eq:clm:thm:parallel_rep_trunc_1_1} is at most $(1-1/n^2)^{m/2}$. Thus the $\AND$ outcome is $0$ with probability at least $1-(1-1/n^2)^{m}$.
\end{itemize}
This completes the proof of \Cref{clm:thm:parallel_rep_trunc_1} by setting the circuit to be $(\OR_n\circ(U^\dag\otimes\mathbb{I}_1))^{\otimes m}$ followed with an $\AND_m$ gate.
\end{proof}

Now we turn to approximating \Cref{eq:thm:parallel_rep_trunc_1}.
For convenience, we use $A\in\{0,1\}^{n\times m}$ to denote a binary matrix of $n$ rows and $m$ columns.
For each $i\in[n]$, we use $A[i]$ to denote the $i$th row of $A$ and use $|A[i]|$ to denote the Hamming weight of $A[i]$.
For $a\in[n]^m$, we associate it with a matrix $A_a\in\{0,1\}^{n\times m}$ by setting the $j$th column as the indicator vector $e_{a_j}$.
Then
\begin{align}
\ket{W_x}^{\otimes m}
&=\left(\frac1{\sqrt n}\sum_i(-1)^{x_i}\ket{e_i}\right)^{\otimes m}
=\frac1{\sqrt{n^m}}\sum_{a\in[n]^m}(-1)^{x_1s_1(a)+\cdots+x_ns_n(a)}\ket{A_a}
\tag{$s_i(a)$ is the number of $i$'s appearance in $a$}\\
&=\frac1{\sqrt{n^m}}\sum_{a\in[n]^m}(-1)^{x_1|A_a[1]|+\cdots+x_n|A_a[n]|}\ket{A_a}
\tag{since $s_i(a)=|A_a[i]|$}\\
&=\frac1{\sqrt{n^m}}\sum_{a\in[n]^m}(-1)^{x_1\cdot\Parity(A_a[1])+\cdots+x_n\cdot\Parity(A_a[n])}\ket{A_a}.
\label{eq:thm:parallel_rep_trunc_2}
\end{align}
Let $h=\polylog(n)$ and define $t=h\cdot m/n=\polylog(n)$.
Now we define symmetric function $f\colon\{0,1\}^m\to\{0,1\}$ by
$$
f(x)=\begin{cases}
\Parity(x) & |x|\le t,\\
0 & \text{otherwise,}
\end{cases}
$$
and define 
\begin{equation}\label{eq:thm:parallel_rep_trunc_3}
\ket{\rho_x}=\frac1{\sqrt{n^m}}\sum_{a\in[n]^m}(-1)^{x_1\cdot f(A_a[1])+\cdots+x_n\cdot f(A_a[n])}\ket{A_a}.
\end{equation}

Then by a simple balls-into-bins calculation, it approximates $\ket{W_x}^{\otimes m}$.

\begin{claim}\label{clm:thm:parallel_rep_trunc_2}
The $\ell_2$ distance between $\ket{\rho_x}$ and $\ket{W_x}^{\otimes m}$ is $2^{-\polylog(n)}$.
\end{claim}
\begin{proof}
Comparing \Cref{eq:thm:parallel_rep_trunc_2} and \Cref{eq:thm:parallel_rep_trunc_3}, we have
\begin{align*}
\left\|\ket{\rho_x}-\ket{W_x}^{\otimes m}\right\|_2^2
&\le4\cdot\Pr_{a\sim[n]^m}\left[|A_a[i]|>t\text{ for some }i\in[n]\right]\\
&\le4n\cdot\Pr_{a\sim[n]^m}\left[|A_a[1]|>t\right]
\tag{by union bound and symmetry}\\
&=4n\cdot\Pr\left[\textsf{binom}(m,1/n)>t\right],
\end{align*}
where $\textsf{binom}(m,1/n)$ is the binomial distribution of $m$ coins with bias $1/n$.
Since $t=h\cdot m/n$ and $h=\polylog(n)$, standard concentration implies that the above probability is $2^{-\polylog(n)}$.
\end{proof}

We also recall from \Cref{thm:symAC0characterization} and \Cref{cor:sym_ac0_in_qac0} that $f$ can be exactly computed with a constant-depth polynomial-size $\QAC$ circuit.
At this point, it suffices to construct $\ket x\ket{\rho_x}$ as follows:
\begin{align*}
\ket x
&\xrightarrow{U^{\otimes m}}\ket{x}\ket{W}^{\otimes m}
\tag{$U$ is the circuit in \Cref{thm:exact_w} preparing $\ket W$}\\
&=\frac1{\sqrt{n^m}}\sum_{a\in[n]^m}\ket x\ket{A_a}
\tag{by the definition of $\ket W$ and $A_a$}\\
&=\frac1{\sqrt{n^m}}\sum_{a\in[n]^m}\bigotimes_{i\in[n]}\left(\ket{x_i}\ket{A_a[i]}\right)
\tag{separating rows of $A_a$}\\
&\xrightarrow{U_f^{\otimes n}}\frac1{\sqrt{n^m}}\sum_{a\in[n]^m}\bigotimes_{i\in[n]}\left(\ket{x_i}\ket{A_a[i]}\ket{f(A_a[i])}\right)
\tag{each $U_f$ evaluates $f(A_a(i))$ and by \Cref{cor:sym_ac0_in_qac0}}\\
&\xrightarrow{R^{\otimes n}}
\frac1{\sqrt{n^m}}\sum_{a\in[n]^m}\bigotimes_{i\in[n]}\left((-1)^{x_i\cdot f(A_a[i])}\ket{x_i}\ket{A_a[i]}\ket{f(A_a[i])}\right)
\tag{$R\colon\ket{u,v}\to(-1)^{u\cdot v}\ket{u,v}$ is a two-qubit unitary}\\
&\xrightarrow{U_f^{\otimes n}}\frac1{\sqrt{n^m}}\sum_{a\in[n]^m}\bigotimes_{i\in[n]}\left((-1)^{x_i\cdot f(A_a[i])}\ket{x_i}\ket{A_a[i]}\right)
\tag{uncompute $U_f^{\otimes n}$}\\
&=\frac1{\sqrt{n^m}}\sum_{a\in[n]^m}(-1)^{x_1\cdot f(A_a[1])+\cdots+x_n\cdot f(A_a[n])}\ket x\ket{A_a}=\ket x\ket{\rho_x}.
\end{align*}
This, combined with \Cref{clm:thm:parallel_rep_trunc_2} and \Cref{clm:thm:parallel_rep_trunc_1}, completes the proof of \Cref{thm:parallel_rep_trunc}.
\end{proof}

At this point, we remark that \emph{if} one can improve the soundness bound $\polylog(n)/n$ in \Cref{thm:parallel_rep_trunc} to constant (or intuitively, achieving $\sim n^2$ parallel runs in $\QAC^0$), then we have $\Parity\in\QAC^0$. 
A direct corollary of \Cref{thm:parallel_rep_trunc} is an improvement of \Cref{cor:general_w_test}, at a negligible sacrifice on the completeness.

\begin{corollary}\label{cor:parallel_rep_trunc}
For every $n$ and $0\le k\le n$, there is a constant-depth polynomial-size $\QAC$ circuit that outputs $1$ with probability at least $1-2^{-\polylog(n)}$ if $\EX_k(x)=1$; and outputs $0$ with probability at least $\polylog(n)/n$ if $\EX_k(x)=0$.

As a consequence, $m=n/\polylog(n)$ input copies suffice for constant-depth polynomial-size $\QAC$ circuits to decide $\EX_k$ with completeness $1-2^{-\polylog(n)}$ and soundness $2^{-\polylog(n)}$.
\end{corollary}

\Cref{cor:parallel_rep_trunc} shows that the copy complexity of every $n$-bit symmetric function is at most $n^2/\polylog(n)$.

\subsection{Probabilistic Computation}

To further reduce copy complexity, we consider probabilistic computation, i.e., a random $\QAC^0$ circuit that correctly computes the target function with high probability.

We illustrate the idea with the parity function.
We use $\Parity^{\uparrow m}\colon(\{0,1\}^n)^m\to\{0,1\}$ to denote the $m$-copy version of the parity function, defined by
$$
\Parity^{\uparrow m}(y_1,\ldots,y_m)=\Parity(y_1)\cdot1_{y_1=\cdots=y_m}
\quad\text{for all $y_1,\ldots,y_c\in(\{0,1\}^n)^m$.}
$$
Recall that \Cref{cor:parallel_rep_trunc} shows $\Parity^{\uparrow n^2/\polylog(n)}\in\QAC^0$. Below we show probabilistic computation yields another square-root saving.

\begin{theorem}
Let $m=n^{1.5}/\polylog(n)$.
There is an ensemble of constant-depth polynomial-size $\QAC$ circuits $\{C_r\}_r$ such that 
$$
\Pr_r\left[C_r(y_1,\ldots,y_m)=\Parity^{\uparrow m}(y_1,\ldots,y_m)\right]\ge1-2^{-\polylog(n)}
\quad\text{holds for every $y_1,\ldots,y_m\in\{0,1\}^n$.}
$$
\end{theorem}
\begin{proof}
Define $f\colon\{0,1\}^n\to\{0,1\}$ by
$$
f(y)=\begin{cases}
\Parity(y) & |y|\in\frac n2\pm\sqrt{n\cdot\polylog(n)},\\
0 & \text{otherwise.}
\end{cases}
$$
Then $f$ is a disjunction of $\sqrt{n\cdot\polylog(n)}$ many exact threshold functions.
By \Cref{cor:parallel_rep_trunc}, with $m=n^{1.5}/\polylog(n)$ copies, constant-depth polynomial-size $\QAC$ circuits can compute $f$ with error $2^{-\polylog(n)}$.

Now for every $r\in\{0,1\}^n$ with even Hamming weight, we use $C_r'$ to denote the $\QAC$ circuit, on $n^{1.5}/\polylog(n)$ copies of $y$, computing $f(y\oplus r)$ with error $2^{-\polylog(n)}$.
We remark that copies of $y$ can be converted into copies of $y\oplus r$ by a layer of single-qubit gates, since $r$ is hardwired into $C_r'$.
In addition, $\Parity(y)=\Parity(y\oplus r)$ since $r$ is an even string, which equals $f(y\oplus r)$ if $|y\oplus r|\in\frac n2\pm\sqrt{n\cdot\polylog(n)}$.

Define $\QAC$ circuit $C_r$ to be $C_r'$ with an additional constant-depth layer to check that the input copies are identical.
Fix arbitrary $y_1,\ldots,y_m\in\{0,1\}^n$ and let $y=y_1$.
If $C_r(y_1,\ldots,y_m)\ne\Parity^{\uparrow m}(y_1,\ldots,y_m)$, then we have the following two cases.
\begin{itemize}
\item $|y\oplus r|\notin\frac n2\pm\sqrt{n\cdot\polylog(n)}$. This happens with probability $2^{-\polylog(n)}$ since $r$ is a uniformly random even string.
\item $|y\oplus r|\in\frac n2\pm\sqrt{n\cdot\polylog(n)}$ but $C_r'$ does not compute $f(y\oplus r)=\Parity(y)$. This happens with probability $2^{-\polylog(n)}$ by \Cref{cor:parallel_rep_trunc}.
\end{itemize}
This completes the proof with a union bound.
\end{proof}

The above construction can be generalized to other symmetric functions (see e.g., \cite{srinivasan2021probabilistic}).
\end{document}